\documentclass[11pt]{article}

\usepackage{graphicx}
\usepackage{amsfonts}
\usepackage{amsmath}
\usepackage{amsthm}
\usepackage[mathcal]{euscript}      % Define tensors
\usepackage{bm}                     % Define vectors
\usepackage{color}
\usepackage{epstopdf}
\usepackage{amssymb}
\usepackage{mathrsfs,mathdots}
\usepackage{lscape}
\usepackage[version=3]{mhchem}

\usepackage[paper=a4paper,dvips,top=3cm,left=2.4cm,right=2.4cm,
    foot=1cm,bottom=3cm]{geometry}

\begin{document}

\title{Spectral Analysis of Piezoelectric Tensors}
\author{Yannan Chen\footnote{%
    School of Mathematics and Statistics, Zhengzhou University, Zhengzhou 450001, China ({\tt ynchen@zzu.edu.cn}).
    This author was supported by the National Natural Science Foundation of China (Grant No. 11401539),
    the Development Foundation for Excellent Youth Scholars of Zhengzhou University (Grant No. 1421315070),
    and the Hong Kong Polytechnic University Postdoctoral Fellowship.}
\and Antal J\'akli\footnote{Liquid Crystal Institute, Kent State University,
     Kent, OH 44242, USA ({\tt ajakli@kent.edu}).},
  \and Liqun Qi\footnote{%
    Department of Applied Mathematics, The Hong Kong Polytechnic University,
    Hung Hom, Kowloon, Hong Kong ({\tt maqilq@polyu.edu.hk}).
    This author's work was partially supported by the Hong Kong Research Grant Council
    (Grant No. PolyU  501913, 15302114, 15300715 and 15301716).}
    }

\date{\today}
\maketitle

\begin{abstract}  A third order real tensor is called a piezoelectric-type tensor
  if it is partially symmetric with respect to its last two indices.
  The piezoelectric tensor is a piezoelectric-type tensor of dimension three.
  We introduce C-eigenvalues and C-eigenvectors for piezoelectric-type tensors.
  Here, ``C'' names after Curie brothers, who first discovered the piezoelectric effect.
  We show that C-eigenvalues always exist, they are invariant under orthonormal transformations,
  and for a piezoelectric-type tensor, the largest C-eigenvalue and its C-eigenvectors form
  the best rank-one piezoelectric-type approximation of that tensor. This means that for the piezoelectric tensor, its largest C-eigenvalue
determines the highest piezoelectric coupling constant.   We further show that for the piezoelectric tensor,
  the largest C-eigenvalue corresponds to the electric displacement vector with the largest $2$-norm in the piezoelectric effect
  under unit uniaxial stress, and the strain tensor with the largest $2$-norm
  in the converse piezoelectric effect under unit electric field vector.
  Thus, C-eigenvalues and C-eigenvectors have concrete physical meanings in
  piezoelectric effect and converse piezoelectric effect.
  Finally, we apply C-eigenvalues and associated C-eigenvectors
  for various piezoelectric crystals with different symmetries.

  \textbf{Key words.} Piezoelectric tensor, piezoelectric effect, converse piezoelectric effect,
  eigenvalue, eigenvector, crystal.
\end{abstract}

\newtheorem{Theorem}{Theorem}[section]
\newtheorem{Definition}[Theorem]{Definition}
\newtheorem{Lemma}[Theorem]{Lemma}
\newtheorem{Corollary}[Theorem]{Corollary}
\newtheorem{Proposition}[Theorem]{Proposition}

% LaTeX definitions
\newcommand{\REAL}{\mathbb{R}}
\newcommand{\COMPLEX}{\mathbb{C}}
\newcommand{\SPHERE}{\mathbb{S}^2}
\newcommand{\diff}{\,\mathrm{d}}
\newcommand{\st}{\mathrm{s.t.}}
\newcommand{\T}{\top}
\newcommand{\vt}[1]{{\bf #1}}%{\bm{#1}}
\newcommand{\x}{\vt{x}}
\newcommand{\y}{\vt{y}}
\newcommand{\z}{\vt{z}}
\newcommand{\Ten}{\mathcal{T}}
\newcommand{\A}{\mathcal{A}}
\newcommand{\B}{\mathcal{B}}
\newcommand{\D}{\mathcal{D}}
\newcommand{\RESULTANT}{\mathrm{Res}}

\newpage
\section{Introduction}

Third order tensors have extensive applications in physics and engineering.   Examples include piezoelectric tensors in crystal study \cite{CC-80, Ha-07, KPG-08, Lo-89, Ny-85, ZTP-13}, third order symmetric traceless-tensors in liquid crystal study \cite{CQV-17, GV-16, Vi-15} and third order susceptibility tensors in nonlinear optics study \cite{Je-70, KGTRU-04}.
Among these third order tensors, the most popular one is the piezoelectric tensor, which plays the key role in piezoelectric effect and converse piezoelectric effect.   Piezoelectricity was discovered by Jacques Curie and Pierre Curie in 1880 \cite{CC-80}.
In the next year, the converse piezoelectric effect was predicted by Lippmann \cite{Li-81}
and confirmed by Curies \cite{CC-81} immediately.
Now it has wide applications in the production and detection of sound, generation of high voltages, electronic frequency
generation, microbalances, and ultra fine focusing of optical assemblies \cite{KPG-08}.

Eigenvalues of higher order tensors were introduced and studied in the recent years \cite{Qi-05, QL-17, QWW-08}.  Particularly, tensor eigenvalues were applied to third order symmetric traceless-tensors in liquid crystal study \cite{CQV-17, GV-16, Vi-15}.    Can we also apply tensor eigenvalues to piezoelectric tensors?    We found that to make it physically meaningful, we may introduce some new eigenvalue definitions for piezoelectric tensors.

In the next section, we introduce C-eigenvalues and C-eigenvectors for piezoelectric-type tensors. Here, ``C'' names after Curie brothers. A third order real tensor is called a piezoelectric-type tensor if it is partially symmetric with respect to its last two indices.   For solid materials, the last two indices of the piezoelectric tensor is symmetric since the stress tensor is symmetric.    Thus, for solid materials, the piezoelectric tensor is a piezoelectric-type tensor of dimension three.  This is not true for liquid crystal,  where there is dissipation \cite{Ja-10, JPSRd-09, JTSSS-02}.  We show that C-eigenvalues always exist, they are invariant under orthonormal transformations, and for a piezoelectric-type tensor, the largest C-eigenvalue and its C-eigenvectors form the best rank-one approximation of that tensor.   This means that for the piezoelectric tensor, its largest C-eigenvalue
determines the highest piezoelectric coupling constant.

In Section 3, we further show that for the piezoelectric tensor in solid crystal, the largest C-eigenvalue corresponds to the electric displacement vector with the largest $2$-norm in the piezoelectric effect under unit uniaxial stress, and the strain tensor with the largest $2$-norm in the converse piezoelectric effect under unit electric field vector.  Thus, C-eigenvalues and C-eigenvectors have concrete physical meanings in piezoelectric effect and converse piezoelectric effect.

In Section 4, we compute C-eigenvalues and associated C-eigenvectors of typical piezoelectric tensors for various crystal classes.

Then, in Section 5, we show that the definition of C-eigenvalues for piezoelectric tensors is different from the definitions of matrix singular values if we look piezoelectric tensors as a $3 \times 6$ matrices \cite{Ny-85, ZTP-13}.

% Some concluding remarks are made in Section 6.
Our results are summarized in Section 6.

\section{Spectral Analysis of Piezoelectric-Type Tensors}

First, we introduce the definition of piezoelectric-type tensors.

\begin{Definition}
  Let $\A=[a_{ijk}]\in\REAL^{n\times n\times n}$ be a third-order $n$ dimensional tensor.
  If the later two indices of $\A$ are symmetric, i.e., $a_{ijk}=a_{ikj}$ for all $j$ and $k$,
  then $\A$ is called a piezoelectric-type tensor.
\end{Definition}

Obviously, the number of independent elements of a piezoelectric-type tensor
$\A\in\REAL^{n\times n\times n}$ is $\frac{1}{2}n^2(n+1)$.
For $\alpha,\beta\in\REAL$ and $\A=[a_{ijk}],\B=[b_{ijk}]\in\REAL^{n\times n\times n}$,
$\alpha\A+\beta\B=[\alpha a_{ijk}+\beta b_{ijk}]\in\REAL^{n\times n\times n}$.
The inner product of $\A$ and $\B$ is $\langle\A,\B\rangle=\sum_{i,j,k}a_{ijk}b_{ijk}$.
A nonnegative scalar $$\|\A\|_F\equiv\sqrt{\langle\A,\A\rangle}$$
is the Frobenius norm of $\A$.
For vectors $\x\in\REAL^n$ and $\y\in\REAL^n$, we denote a scalar
\begin{equation*}
    \x\A\y\y \equiv \sum_{i,j,k}a_{ijk}x_iy_jy_k \in\REAL
\end{equation*}
as a product of the piezoelectric tensor $\A$ with vectors $\x$ and $\y$.
Moreover, we define
\begin{equation*}
    \A\y\y \equiv \left(
               \begin{array}{c}
                 \sum_{j,k}a_{1jk}y_jy_k \\
                 \vdots \\
                 \sum_{j,k}a_{njk}y_jy_k \\
               \end{array}
             \right)\in\REAL^n \quad\text{and}\quad
    \x\A\y \equiv \left(
               \begin{array}{c}
                 \sum_{i,j}x_ia_{ij1}y_j \\
                 \vdots \\
                 \sum_{i,j}x_ia_{ijn}y_j \\
               \end{array}
             \right)\in\REAL^n.
\end{equation*}

Let $\lambda$ be a real number and $\x,\y\in\REAL^n$ be unit vectors, i.e, $\x^\top \x = 1$ and $\y^\top \y = 1$.
Elements of a rank-one piezoelectric-type tensor $\lambda \x\circ\y\circ\y\in\REAL^{n\times n\times n}$ are
$[\lambda \x\circ\y\circ\y]_{ijk} = \lambda x_iy_jy_k$
for $i,j,k=1,2,\ldots,n$. Here, ``$\circ$'' means the outer product.
If a scalar $\lambda \in \REAL$ and vectors $\x,\y\in\REAL^n$ minimize the following optimization problem
\begin{equation} \label{rank-1}
    \min \left\{ \|\A - \lambda \x\circ\y\circ\y\|_F^2 : \lambda \in \REAL, \x^\top \x = 1, \y^\top \y = 1 \right\},
\end{equation}
then $\lambda \x\circ\y\circ\y$ is called the best rank-one piezoelectric-type approximation of $\A$.

Using these notations, we present the following definition of C-eigenvalues and
C-eigenvectors of a piezoelectric tensor. Here, ``C'' names after Curie brothers.

\begin{Definition}
  Let $\A\in\REAL^{n\times n\times n}$ be a piezoelectric-type tensor.
  If there exist a scalar $\lambda\in\REAL$, vectors $\x\in\REAL^n$ and $\y\in\REAL^n$ satisfying
  the following system
  \begin{equation}\label{C-eig}
    \A\y\y = \lambda\x, \quad \x\A\y = \lambda\y, \quad \x^\T\x=1, \quad\text{ and }\quad \y^\T\y=1,
  \end{equation}
  then, $\lambda$ is called a C-eigenvalue of $\A$,
  $\x$ and $\y$ are called associated left and right C-eigenvectors, respectively.
\end{Definition}

We have the following theorem.

\begin{Theorem}\label{Th-2.3}
Let $\A$ be a piezoelectric-type tensor.
Then we have the following conclusions.

\medskip

(a) There exist C-eigenvalues of $\A$ and associated left and right C-eigenvectors.

\medskip

(b) Suppose that $\lambda$, $\x$ and $\y$ are  a C-eigenvalue and its associated left and right C-eigenvectors of $\A$, respectively. Then $$\lambda = \x\A\y\y.$$
  Furthermore, $(\lambda,\x,-\y)$, $(-\lambda,-\x,\y)$, and $(-\lambda,-\x,-\y)$
  are also C-eigenvalues and their associated C-eigenvectors of $\A$.
\medskip

(c) Denote the largest C-eigenvalue of $\A$ and its associated left and right C-eigenvectors as
$\lambda^*$, $\x^*$ and $\y^*$, respectively.  Then
\begin{equation}\label{maximum}
  \lambda^* = \max~\{\x\A\y\y : \x^\top \x =1, \y^\top \y = 1\}.
\end{equation}
Furthermore, $\lambda^*\x^*\circ\y^*\circ\y^*$ forms the best rank-one piezoelectric-type approximation of $\A$.
\end{Theorem}
\begin{proof}
  (a) We consider the following optimization problem
  \begin{equation}\label{opt-1}
    \max~\{\x\A\y\y : \x^\T\x=1, \y^\T\y=1\}.
  \end{equation}
  On one hand, since the objective function $\x\A\y\y$ is continuous in variables $\x$ and $\y$ and
  the feasible region $\{(\x,\y)\in\REAL^n\times\REAL^n:\x^\T\x=1,\y^\T\y=1\}$ is compact,
  there exist vectors $\x^*$ and $\y^*$ that solve \eqref{opt-1} with
  the maximal objective value $\lambda^*\equiv\x^*\A\y^*\y^*$.

  On the other hand, we write the Lagrangian of \eqref{opt-1}:
  \begin{equation}\label{opt-2}
    L(\x,\y,\mu_1,\mu_2) = -\x\A\y\y + \frac{\mu_1}{2}(\x^\T\x-1) + \mu_2(\y^\T\y-1).
  \end{equation}
  By the Lagrangian multiplier method, for the optimal solution $(\x^*,\y^*)$,
  there exist multipliers $\mu_1^*$ and $\mu_2^*$ such that
  \begin{equation*}
  \left\{\begin{array}{rcl}
    \frac{\partial L}{\partial \x^*} =&  -\A\y^*\y^* +  \mu_1^*\x^* &=0, \\
    \frac{\partial L}{\partial \y^*} =& -2\x^*\A\y^* + 2\mu_2^*\y^* &=0, \\
    \frac{\partial L}{\partial \mu_1^*} =& \frac{1}{2}({\x^*}^\T\x^*-1) &=0, \\
    \frac{\partial L}{\partial \mu_2^*} =& {\y^*}^\T\y^*-1 &=0.
  \end{array}\right.
  \end{equation*}
  By ${\x^*}^\T\x^*={\y^*}^\T\y^*=1$, we have $\mu_1^*=\mu_2^*=\x^*\A\y^*\y^*=\lambda^*$.
  Hence, $\lambda^*$, $\x^*$ and $\y^*$ satisfy \eqref{C-eig} and hence
  a C-eigenvalue and its associated left and right C-eigenvectors of $\A$.
  This proves the existence.

  (b) It is straightforward to verify the assertion (b).

  (c) By some calculations, we find
  \begin{equation*}
    \|\A-\lambda\x\circ\y\circ\y\|_F^2
      = \|A\|_F^2 - 2\lambda\langle\A,\x\circ\y\circ\y\rangle + \lambda^2\|\x\|^2\|\y\|^4.
  \end{equation*}
  Minimizing this square-cost with respect to $\lambda$, we get $\lambda=\langle\A,\x\circ\y\circ\y\rangle$
  because of $\|\x\|=\|\y\|=1$. By substituting $\lambda=\langle\A,\x\circ\y\circ\y\rangle$
  to the square-cost, we have
  \begin{equation*}
    \|\A-\lambda\x\circ\y\circ\y\|_F^2
      = \|\A\|^2 - \langle\A,\x\circ\y\circ\y\rangle^2.
  \end{equation*}
  Hence, there is a dual problem of \eqref{rank-1} \cite{ZG-01}:
  \begin{equation*}
    \max~\{\langle\A,\x\circ\y\circ\y\rangle : \x^\T\x =1, \y^\T\y =1 \}.
  \end{equation*}
  Then, there exist vectors $\x^*$ and $\y^*$ such that
  \begin{eqnarray*}
    \langle\A,\x^*\circ\y^*\circ\y^*\rangle &=& \max~\{\langle\A,\x\circ\y\circ\y\rangle : \x^\T\x =1, \y^\T\y =1\}
  \end{eqnarray*}
  because a piezoelectric-type tensor $\A$ is partially symmetric with respect to the later two indices.
  This yields \eqref{rank-1}.
  Here, $\langle\A,\x^*\circ\y^*\circ\y^*\rangle = \x^*\A\y^*\y^*=\lambda^*$ is the largest C-eigenvalue of $\A$,
  $\x^*$ and $\y^*$ are its associated left and right C-eigenvectors respectively.
\end{proof}

By Theorem 2.3 (c), $\lambda^*\x^*\circ\y^*\circ\y^*$ forms the best rank-one
piezoelectric type approximation of $\A$. This implies that for the piezoelectric tensor $\A$, its largest C-eigenvalue
determines the highest piezoelectric coupling constant, and
$\y^*$ is the corresponding direction of the stress where this appears.   Thus, the largest C-eigenvalue of the piezoelectric tensor has concrete physical meaning.  In the next section, we will further
discuss its meanings.

\medskip

Next, we show that C-eigenvalues of a piezoelectric-type tensor $A=[a_{ijk}]$ are
invariant under orthogonal transformations.
Let $Q=[q_{ir}]\in\REAL^{n\times n}$ be an orthogonal matrix.
We define a new tensor $\A Q^3\in\REAL^{n\times n\times n}$ in which
elements are
\begin{equation*}
    [\A Q^3]_{rst} = \sum_{i,j,k}a_{ijk}q_{ir}q_{js}q_{kt}
\end{equation*}
for $r,s,t=1,2,\ldots,n$. Obviously, $\A Q^3$ is also a piezoelectric-type tensor.

\begin{Theorem}
  Suppose that $Q\in\REAL^{n\times n}$ is an orthogonal matrix.
  Let $\lambda$, $\x$ and $\y$ be a C-eigenvalue and its associated C-eigenvectors of a piezoelectric-type tensor $\A\in\REAL^{n\times n\times n}$.
  Then, $\lambda$, $Q^\T\x$ and $Q^\T\y$ are a C-eigenvalue and its associated C-eigenvectors of $\A Q^3$.
\end{Theorem}
\begin{proof}
  We look at the $r$th component of a vector $(\A Q^3)(Q^\T\y)(Q^\T\y)$:
  \begin{eqnarray*}
    \lefteqn{[(\A Q^3)(Q^\T\y)(Q^\T\y)]_r = \sum_{s,t}[\A Q^3]_{rst}[Q^\T\y]_s[Q^\T\y]_t} \\
      &=& \sum_{s,t}\left(\sum_{i,j,k}a_{ijk}q_{ir}q_{js}q_{kt}\right)\left(\sum_{\hat{j}}q_{\hat{j}s}y_{\hat{j}}\right)\left(\sum_{\hat{k}}q_{\hat{k}t}y_{\hat{k}}\right) \\
      &=& \sum_{i}q_{ir}\left(\sum_{j,k}a_{ijk}\left(\sum_{s,\hat{j}}q_{js}q_{\hat{j}s}y_{\hat{j}}\right)\left(\sum_{t,\hat{k}}q_{kt}q_{\hat{k}t}y_{\hat{k}}\right)\right) \\
      &=& \sum_{i}q_{ir}\left(\sum_{j,k}a_{ijk}[QQ^\T\y]_j[QQ^\T\y]_k\right) \\
      &=& \sum_{i}q_{ir}\left(\sum_{j,k}a_{ijk}y_jy_k\right) \\
      &=& \sum_iq_{ir}[\A\y\y]_i,
  \end{eqnarray*}
  where the orthogonal matrix $Q$ satisfies $QQ^\T=I$.
  Then, we get $(\A Q^3)(Q^\T\y)(Q^\T\y) = Q^\T(\A\y\y)$. From $\A\y\y=\lambda\x$,
  we immediately get an equation $$(\A Q^3)(Q^\T\y)(Q^\T\y) = \lambda Q^\T\x.$$
  Similarly, we have $$(Q^\T\x)(\A Q^3)(Q^\T\y) = \lambda Q^\T\y.$$
  In addition, $$(Q^\T\x)^\T (Q^\T\x) = (Q^\T\y)^\T (Q^\T\y) =1.$$
  Hence, $\lambda$ is a C-eigenvalue of $\A Q^3$, $Q^\T\x$ and $Q^\T\y$ are its associated C-eigenvectors.
\end{proof}

\section{Applications in Piezoelectric Effect and Converse Piezoelectric Effect}

In the last section, we showed that for the piezoelectric tensor, the largest C-eigenvalue of $\A$
determines the highest piezoelectric coupling constant, and
$\y^*$ is the corresponding direction of the stress where this appears.   We now further
discuss its physical meanings.

For non-centrosymmetric materials, the linear piezoelectric equation is expressed as
\begin{equation*}
    P_i = \sum_{j,k}a_{ijk}T_{jk},
\end{equation*}
where $\A=[a_{ijk}]\in\REAL^{3\times3\times3}$ is a piezoelectric tensor,
$T\in\REAL^{3\times3}$ is the stress tensor,
and $P$ is the electric change density displacement (polarization).
Since $\A$ is a piezoelectric-type tensor,
the last two indices of $\A$ is symmetric,
i.e., $a_{ijk}=a_{ikj}$ for all $j$ and $k$. Hence, there are $18$
independent elements in $\A$.

What situations trigger the extreme piezoelectricity under unit uniaxial stress?
An example of uniaxial stress is the stress in a long, vertical rod loaded by hanging a weight on the end \cite[Page 90]{Ny-85}.
In this case, the stress tensor could be rewritten as $T=\y\y^{\T}$ with $\y^\top \y =1$.
Then, we consider the following problem
\begin{equation}\label{lonEff-1}
\left\{\begin{aligned}
    \max~~ & \|P\|_2 \\
    \st~~~ & P = \A\y\y, \\
           & \y^\T\y=1.
\end{aligned}\right.
\end{equation}
Using a dual norm, we have $\|P\|_2=\max_{\x^\T\x=1} \x^\T P=\max_{\x^\T\x=1}\x\A\y\y$.
Hence, it suffices to consider the optimization problem
\begin{equation}\label{piezo-opt}
    \max ~ \x\A\y\y \qquad \st~~\x^\T \x=1,~\y^\T\y=1.
\end{equation}
We denote $(\x^*,\y^*)$ as the optimal solution of the above optimization problem.
Then, $\lambda^*=\x^*\A\y^*\y^*$ is the largest C-eigenvalue of the piezoelectric tensor $\A$,
and $\y^*$ is the unit uniaxial direction that the extreme piezoelectric effect along took place.
Then we have the following theorem.

\begin{Theorem}
  Suppose that $\lambda^*$, $\x^*$ and $\y^*$ are the largest C-eigenvalue and
   its associated C-eigenvectors of the piezoelectric tensor $\A$.
  Then, $\lambda^*$ is the maximum value of the $2$-norm of the electric polarization
  under a unit uniaxial stress along direction $\y^*$.
\end{Theorem}

\medskip

The linear equation for the converse piezoelectric effect is
\begin{equation*}
    S_{jk}=\sum_{i}a_{ijk}E_i,
\end{equation*}
where $S$ is the strain tensor and $E$ is the electric field strength.
Let $\|\cdot\|_2$ be the matrix spectral norm, i.e.,
$\|S\|_2=\max_{\y^{\T}\y=1}\y^\T S\y.$
Now, we maximize the spectral norm of $S$:
\begin{equation}\label{conPiezo}
\left\{\begin{aligned}
    \max~~ & \|S\|_2 \\
    \st~~~ & S_{jk}=\sum_{i=1}^3 E_ia_{ijk} \qquad\forall j,k\in\{1,2,3\},\\
           & \|E\| = 1.
\end{aligned}\right.
\end{equation}
Since $\|S\|_2=\max_{\y^{\T}\y=1}\y^\T S\y = \max_{\y^{\T}\y=1} E\A\y\y$,
we rewrite \eqref{conPiezo} as follows
$$ \max~\{E\A\y\y : E^\T E=1, \y^\T\y=1\}. $$
We denote $(E^*,\y^*)$ as the optimal solution of the above optimization problem.
Then, $\lambda^*=E^*\A\y^*\y^*$ is the largest C-eigenvalue of $\A$,
$E^*$ and $\y^*$ are its associated left and right C-eigenvectors.

\begin{Theorem}
Suppose that $\lambda^*$, $\x^*$ and $\y^*$ are the largest C-eigenvalue and
its associated C-eigenvectors of the piezoelectric tensor $\A$.
  Then, $\lambda^*$ is the largest spectral norm of a strain tensor generated by
  the converse piezoelectric effect under unit electric field strength $\|\x^*\|=1$.
\end{Theorem}

\section{C-Eigenvalues for Piezoelectric Tensors}

Owing to the crystallographic symmetry of materials, there are $32$ classes in crystals \cite{Ha-07}.
However, for $11$ classes in crystals possing the center of symmetry, piezoelectricity vanishes.
For the class $432$, piezoelectric changes cancel each other. Hence,
piezoelectricity may exist in the remaining $20$ crystallographic classes.
We examine some typical crystals in these classes in this section.

By Theorem \ref{Th-2.3} (b), we know that $(-\lambda,-\x,\y)$, $(\lambda,\x,-\y)$,
and $(-\lambda,-\x,-\y)$ are C-eigenvalues and associated C-eigenvectors of a piezoelectric tensor
if $(\lambda,\x,\y)$ are a C-eigenvalue and associated C-eigenvectors of the piezoelectric tensor.
In this section, we use $(\lambda,\x,\y)$ to present a group of these four C-eigenvalues and
associated C-eigenvectors of the piezoelectric tensor for compactness.

Piezoelectric tensors of crystals have special structures owing to the symmetry.
First, we consider crystals in $23$ and $\bar{4}3m$ crystallographic point groups.
There is only one independent parameter in the corresponding piezoelectric tensor $\A(\alpha)$:
\begin{equation*}
    a_{123}=a_{213}=a_{312}=-\alpha,
\end{equation*}
where $\alpha\neq 0$. Other elements of $\A(\alpha)$ are zeros.
Then, we have the following proposition.

\begin{Proposition}\label{Pr-4.1}
  There are $13$ groups of C-eigenvalues and associated C-eigenvectors of $\A(\alpha)$.
\end{Proposition}
\begin{proof}
  We solve the polynomial system \eqref{C-eig} for $\A(\alpha)$:
  \begin{subequations}
  \begin{align}
    -2\alpha y_2y_3 =& \lambda x_1, \label{9a}\\
    -2\alpha y_1y_3 =& \lambda x_2, \label{9b}\\
    -2\alpha y_1y_2 =& \lambda x_3, \label{9c}\\
    -\alpha x_2y_3 - \alpha x_3y_2 =& \lambda y_1, \label{9d}\\
    -\alpha x_1y_3 - \alpha x_3y_1 =& \lambda y_2, \label{9e}\\
    -\alpha x_1y_2 - \alpha x_2y_1 =& \lambda y_3, \label{9f}\\
    x_1^2 + x_2^2 + x_3^2 =& 1, \label{9g}\\
    y_1^2 + y_2^2 + y_3^2 =& 1. \label{9h}
  \end{align}
  \end{subequations}
  If $\lambda=0$, we know two of $y_1,y_2,y_3$ are zeros from \eqref{9a}-\eqref{9c} and \eqref{9h}.
  Assume $\y=(\pm 1,0,0)^\T$. By \eqref{9e}-\eqref{9f}, we have $x_2=x_3=0$.
  In addition, $x_1=\pm 1$ by \eqref{9g}. Hence, we get the solution
  \begin{equation*}
    \lambda_1^* = 0, \quad \x_1^* = (1,0,0)^\T, \quad\y_1^*=(1,0,0)^\T.
  \end{equation*}
  Similarly, we have
  \begin{eqnarray*}
    \lambda_2^* = 0, \quad& \x_2^* = (0,1,0)^\T, &\quad \y_2^*=(0,1,0)^\T, \\
    \lambda_3^* = 0, \quad& \x_3^* = (0,0,1)^\T, &\quad \y_3^*=(0,0,1)^\T.
  \end{eqnarray*}

  Next, we consider the case that $\lambda\neq 0$ and $y_3=0$.
  By \eqref{9a}-\eqref{9b}, we get $x_1=x_2=0$.
  In addition, by \eqref{9g}, we obtain $x_3=\pm 1$.
  From \eqref{9d}-\eqref{9e}, we know $\mp \alpha y_2=\lambda y_1$ and $\mp \alpha y_1=\lambda y_2$.
  Then, $\lambda y_1y_2 = \mp \alpha y_2^2 = \mp \alpha y_1^2$. Hence, $y_1^2=y_2^2=\frac{1}{2}$ by \eqref{9h}.
  In short, we have two solutions
  \begin{equation*}
    \lambda_4^* = \alpha, \quad \x_4^* = (0,0,-1)^\T,
    \quad \y_4^*=\left(\frac{1}{\sqrt{2}},\frac{1}{\sqrt{2}},0\right)^\T,
  \end{equation*}
  and
  \begin{equation*}
    \lambda_5^* = \alpha, \quad \x_5^* = (0,0,1)^\T,
    \quad \y_5^*=\left(\frac{1}{\sqrt{2}},-\frac{1}{\sqrt{2}},0\right)^\T.
  \end{equation*}
  Using a similar approach, we obtain four more solutions for the case of $\lambda\neq 0$ and $y_2=0$
  and the case of $\lambda\neq 0$ and $y_1=0$:
  \begin{eqnarray*}
    \lambda_6^* = \alpha, \quad& \x_6^* = (0,-1,0)^\T, &
    \quad \y_6^*=\left(\frac{1}{\sqrt{2}},0,\frac{1}{\sqrt{2}}\right)^\T, \\
    \lambda_7^* = \alpha, \quad& \x_7^* = (0,1,0)^\T, &
    \quad \y_7^*=\left(\frac{1}{\sqrt{2}},0,-\frac{1}{\sqrt{2}}\right)^\T, \\
    \lambda_8^* = \alpha, \quad& \x_8^* = (-1,0,0)^\T, &
    \quad \y_8^*=\left(0,\frac{1}{\sqrt{2}},\frac{1}{\sqrt{2}}\right)^\T, \\
    \lambda_9^* = \alpha, \quad& \x_9^* = (1,0,0)^\T, &
    \quad \y_9^*=\left(0,\frac{1}{\sqrt{2}},-\frac{1}{\sqrt{2}}\right)^\T.
  \end{eqnarray*}

  Finally, we consider the case that $y_1\neq 0$, $y_2\neq 0$, $y_3\neq 0$, and $\lambda\neq 0$.
  From \eqref{9a}-\eqref{9c}, we know $x_1\neq 0$, $x_2\neq 0$, and $x_3\neq 0$.
  By \eqref{9a}-\eqref{9b}, we have $$\frac{x_1}{x_2}=\frac{y_2}{y_1}\equiv t\neq0.$$
  Then, by multiplying $t$ to both sides of \eqref{9d}, we get
  $ -\alpha x_1 y_3 - \alpha x_3y_1 t^2 = \lambda y_2.$
  Combining this equation and \eqref{9e}, we get $-\alpha x_3y_1(t^2-1)=0$.
  Hence,  $$t^2=1, \quad x_1^2=x_2^2, \quad\text{ and }\quad y_1^2=y_2^2. $$
  Similarly, we have $x_1^2=x_3^2$ and $y_1^2=y_3^2$.
  By \eqref{9g} and $\eqref{9h}$, we have $x_1^2=x_2^2=x_3^2=\frac{1}{3}$
  and $y_1^2=y_2^2=y_3^2=\frac{1}{3}$. In a word, we obtain four more solutions
  \begin{eqnarray*}
    \lambda_{10}^* = \frac{2\alpha}{\sqrt{3}},
    \quad& \x_{10}^* = \left(-\frac{1}{\sqrt{3}},-\frac{1}{\sqrt{3}},-\frac{1}{\sqrt{3}}\right)^\T, &
    \quad \y_{10}^*=\left(\frac{1}{\sqrt{3}},\frac{1}{\sqrt{3}},\frac{1}{\sqrt{3}}\right)^\T, \\
        \lambda_{11}^* = \frac{2\alpha}{\sqrt{3}},
    \quad& \x_{11}^* = \left(-\frac{1}{\sqrt{3}},\frac{1}{\sqrt{3}},-\frac{1}{\sqrt{3}}\right)^\T, &
    \quad \y_{11}^*=\left(\frac{1}{\sqrt{3}},-\frac{1}{\sqrt{3}},\frac{1}{\sqrt{3}}\right)^\T, \\
        \lambda_{12}^* = \frac{2\alpha}{\sqrt{3}},
    \quad& \x_{12}^* = \left(-\frac{1}{\sqrt{3}},-\frac{1}{\sqrt{3}},\frac{1}{\sqrt{3}}\right)^\T, &
    \quad \y_{12}^*=\left(\frac{1}{\sqrt{3}},\frac{1}{\sqrt{3}},-\frac{1}{\sqrt{3}}\right)^\T, \\
        \lambda_{13}^* = \frac{2\alpha}{\sqrt{3}},
    \quad& \x_{13}^* = \left(-\frac{1}{\sqrt{3}},\frac{1}{\sqrt{3}},\frac{1}{\sqrt{3}}\right)^\T, &
    \quad \y_{13}^*=\left(\frac{1}{\sqrt{3}},-\frac{1}{\sqrt{3}},-\frac{1}{\sqrt{3}}\right)^\T.
  \end{eqnarray*}
  The proof is complete.
\end{proof}

Now, we are presenting C-eigenvalues and associated C-eigenvectors of
piezoelectric tensors arising from known piezoelectric materials with different symmetries.
Here, coefficients of piezoelectric tensors are measured in \cite{dCGAP-15}
with unit {(pC/N)} that is omitted for convenience.
We note that ``p'' means pico ($10^{-12}$), ``C'' stands for coulomb (electric),
and ``N'' is newton (force).

\medskip

\textbf{Example 1.}
The compound \ce{VFeSb} belongs to the $\bar{4}3m$ crystallographic point group \cite{dCGAP-15}.
Nonzero coefficients of the piezoelectric tensor $\A_{VFeSb}$ are
\begin{equation*}
    a_{123}=a_{213}=a_{312}=-3.68180667.
\end{equation*}
By Proposition \ref{Pr-4.1}, we find that the largest C-eigenvalue of $\A_{VFeSb}$
is about $4.25138$.

\medskip

In remainder examples, the polynomial system \eqref{C-eig} of C-eigenvalues are solved by
the function \texttt{NSolve} in Mathematica.

\medskip

\textbf{Example 2.}
The piezoelectricity of $\alpha$-quartz (\ce{SiO2}) crystal were discovered by
Curie brothers in 1880 \cite{CC-80}.
The $\alpha$-quartz belongs to the 32 crystallographic point group \cite{Ha-07}.
Hence, there are two independent parameters in the piezoelectric tensor $\A_{SiO2}$:
$$ a_{111}=-a_{122}=-a_{212}=-0.13685  \quad\text{ and }
\quad a_{123}=-a_{213}=-0.009715.$$
Other elements of the piezoelectric tensor are zeros.
Positive C-eigenvalue of $\A_{SiO2}$
and associated C-eigenvectors are reported in Table \ref{C-eigen-SiO2}.

\begin{table}[tbph]
  \begin{center}
  \caption{Positive C-eigenvalues of a piezoelectric tensor of $\alpha$-quartz.}\label{C-eigen-SiO2}
  \begin{tabular}{c|c|rrr|rrr}
    \hline
    No. & $\lambda$ & \multicolumn{3}{c|}{$\x^\T$} & \multicolumn{3}{c}{$\y^\T$}  \\ \hline
      1 &0.137536&1.0&0.0&0.0&0.0&0.997515&-0.0704604 \\
      2 &0.137536&-0.5&0.866025&0.0&0.863873&0.498757&0.0704604 \\
      3 &0.137536&-0.5&-0.866025&0.0&0.863873&-0.498757&-0.0704604 \\
      4 &0.13685&-1.0&0.0&0.0&1.0&0.0&0.0 \\
      5 &0.13685&0.5&0.866025&0.0&0.5&0.866025&0.0 \\
      6 &0.13685&0.5&-0.866025&0.0&0.5&-0.866025&0.0 \\
      7 &0.000686228&-1.0&0.0&0.0&0.0&0.0704604&0.997515 \\
      8 &0.000686228&0.5&-0.866025&0.0&0.0610205&0.0352302&-0.997515 \\
      9 &0.000686228&0.5&0.866025&0.0&0.0610205&-0.0352302&0.997515 \\
    \hline
  \end{tabular}
  \end{center}
\end{table}

\textbf{Example 3.}
The compound \ce{Cr2AgBiO8} belongs to the $\bar{4}$ crystallographic point group \cite{dCGAP-15}.
There are four independent parameters in the piezoelectric tensor $\A_{Cr2AgBiO8}$
\begin{eqnarray*}
  & a_{123}=a_{213}=-0.22163, \quad a_{113}=-a_{223}=2.608665, & \\
    & a_{311}=-a_{322}=0.152485, \quad\text{ and }\quad a_{312}=-0.37153. &
\end{eqnarray*}
Other elements of the piezoelectric tensor are zeros.
Positive C-eigenvalue of $\A_{Cr2AgBiO8}$
and associated C-eigenvectors are reported in Table \ref{C-eigen-Cr2AgBiO8}.

\begin{table}[tbph]
  \begin{center}
  \caption{Positive C-eigenvalues of the piezoelectric tensor of \ce{Cr2AgBiO8}.}\label{C-eigen-Cr2AgBiO8}
  \begin{tabular}{c|c|rrr|rrr}
    \hline
    No. & $\lambda$ & \multicolumn{3}{c|}{$\x^\T$} & \multicolumn{3}{c}{$\y^\T$}  \\ \hline
      1 & 2.6258 & 0.872141&0.48317&0.0769254 & 0.589036&-0.394962&0.705012 \\
      2 & 2.6258 & -0.872141&-0.48317&0.0769254 & 0.589036&-0.394962&-0.705012 \\
      3 & 2.6258 & 0.48317&-0.872141&-0.0769254 & 0.394962&0.589036&0.705012 \\
      4 & 2.6258 & -0.48317&0.872141&-0.0769254 & 0.394962&0.589036&-0.705012 \\
      5 & 2.61806 & 0.961197&-0.275862&0.0 & 0.693742&0.136827&0.707107 \\
      6 & 2.61806 & -0.961197&0.275862&0.0 & 0.693742&0.136827&-0.707107 \\
      7 & 2.61806 & 0.275862&0.961197&0.0 & 0.136827&-0.693742&0.707107 \\
      8 & 2.61806 & -0.275862&-0.961197&0.0 & 0.136827&-0.693742&-0.707107 \\
      9 & 0.401605 & 0.0&0.0&1.0 & 0.830569&-0.556916&0.0 \\
     10 & 0.401605 & 0.0&0.0&-1.0 & 0.556916&0.830569&0.0 \\
    \hline
  \end{tabular}
  \end{center}
\end{table}

\textbf{Example 4.}
The compound \ce{RbTaO3} belongs to the $3m$ crystallographic point group \cite{dCGAP-15}.
There are four independent parameters in the piezoelectric tensor $\A_{RbTaO3}$
\begin{eqnarray*}
    & a_{113}=a_{223}=-8.40955, \quad a_{222}=-a_{212}=-a_{211}=-5.412525 & \\
    & a_{311}=a_{322}=-4.3031, \quad\text{ and }\quad a_{333}=-5.14766. &
\end{eqnarray*}
Other elements of the piezoelectric tensor are zeros.
Positive C-eigenvalue of $\A_{RbTaO3}$
and associated C-eigenvectors are reported in Table \ref{C-eigen-RbTaO3}.

\begin{table}[!tbph]
  \begin{center}
  \caption{Positive C-eigenvalues of the piezoelectric tensor of \ce{RbTaO3}.}\label{C-eigen-RbTaO3}
  \begin{tabular}{c|c|rrr|rrr}
    \hline
    No. & $\lambda$ & \multicolumn{3}{c|}{$\x^\T$} & \multicolumn{3}{c}{$\y^\T$}  \\ \hline
      1 & 12.4234 & 0.804378&0.464408&-0.370541 & 0.695227&0.401389&-0.596277 \\
      2 & 12.4234 & -0.804378&0.464408&-0.370541 & 0.695227&-0.401389&0.596277 \\
      3 & 12.4234 & 0.0&-0.928816&-0.370541 & 0.0&0.802779&0.596277 \\
      4 & 7.82245 & 0.677808&-0.391333&-0.622442 & 0.49743&-0.287191&-0.818587 \\
      5 & 7.82245 & -0.677808&-0.391333&-0.622442 & 0.49743&0.287191&0.818587 \\
      6 & 7.82245 & 0.0&0.782666&-0.622442 & 0.0&0.574382&-0.818587 \\
      7 & 6.91463 & 0.677894&-0.391382&-0.622318 & 0.5&0.866025&0.0 \\
      8 & 6.91463 & -0.677894&-0.391382&-0.622318 & 0.5&-0.866025&0.0 \\
      9 & 6.91463 & 0.0&0.782764&-0.622318 & 1.0&0.0&0.0 \\
     10 & 5.14766 & 0.0&0.0&-1.0 & 0.0&0.0&1.0 \\
     11 & 4.38052 & 0.0247105&0.0142666&-0.999593 & 0.826334&0.477084&0.299271 \\
     12 & 4.38052 & -0.0247105&0.0142666&-0.999593 & 0.826334&-0.477084&-0.299271 \\
     13 & 4.38052 & 0.0&-0.0285332&-0.999593 & 0.0&0.954168&-0.299271 \\
    \hline
  \end{tabular}
  \end{center}
\end{table}

%\ce{Na2O} belongs to the $3$ crystallographic point group.
%There are six independent parameters in the piezoelectric tensor $\A_{Na2O}$
%\begin{eqnarray*}
%    & a_{111}=-a_{122}=-a_{212}=0.80022, \quad a_{123}=-a_{213}=0.138975, \quad a_{113}=a_{223}=-0.16963, & \\
%    & a_{222}=-a_{112}==-a_{211}=-0.043735, \quad a_{311}=a_{322}=-2.339395, \quad\text{ and } a_{333}=-9.61902. &
%\end{eqnarray*}
%Other elements of the piezoelectric tensor is zero.
%Positive C-eigenvalue of $\A_{Na2O}$
%and associated C-eigenvectors are reported in Table \ref{C-eigen-Na2O}.
%
%\begin{table}[bph]
%  \begin{center}
%  \caption{C-eigenvalues and associated C-eigenvectors of
%  the piezoelectric tensor of $\alpha$-quartz.}\label{C-eigen-Na2O}
%  \begin{tabular}{c|c|rrr|rrr}
%    \hline
%    No. & $\lambda$ & \multicolumn{3}{c|}{$\x^\T$} & \multicolumn{3}{c}{$\y^\T$}  \\ \hline
%      1 & 9.61902 & 0.&0.&-1. & 0.&0.&1. \\
%      2 & 2.47286 & 0.2726&-0.17527&-0.946028 & 0.951413&0.307917&0. \\
%      3 & 2.47286 & 0.0154875&0.323714&-0.946028 & 0.74237&-0.66999&0. \\
%      4 & 2.47286 & -0.288088&-0.148444&-0.946028 & 0.209043&0.977906&0. \\
%      5 & 2.4721 & 0.286464&0.147607&-0.946652 & 0.977851&-0.209031&0.0106522 \\
%      6 & 2.4721 & -0.271064&0.174281&-0.946652 & 0.307899&-0.951359&-0.0106521 \\
%      7 & 2.4721 & -0.0153997&-0.321888&-0.946652 & 0.669952&0.742327&-0.0106521 \\
%    \hline
%  \end{tabular}
%  \end{center}
%\end{table}

\textbf{Example 5.}
The compound \ce{NaBiS2} belongs to the $mm2$ crystallographic point group \cite{dCGAP-15}.
There are five independent parameters in the piezoelectric tensor $\A_{NaBiS2}$
\begin{eqnarray*}
    & a_{113}=-8.90808, \quad a_{223}=-0.00842, \quad a_{311}=-7.11526, & \\
    & a_{322}=-0.6222, \quad\text{ and }\quad a_{333}=-7.93831. &
\end{eqnarray*}
Other elements of the piezoelectric tensor are zeros.
Positive C-eigenvalue of $\A_{NaBiS2}$
and associated C-eigenvectors are reported in Table \ref{C-eigen-NaBiS2}.

\begin{table}[tbph]
  \begin{center}
  \caption{Positive C-eigenvalues of the piezoelectric tensor of \ce{NaBiS2}.}\label{C-eigen-NaBiS2}
  \begin{tabular}{c|c|rrr|rrr}
    \hline
    No. & $\lambda$ & \multicolumn{3}{c|}{$\x^\T$} & \multicolumn{3}{c}{$\y^\T$}  \\ \hline
     1 & 11.6674 & 0.762919&0.0&-0.646494 & 0.693139&0.0&-0.720804 \\
     2 & 11.6674 & -0.762919&0.0&-0.646494 & 0.693139&0.0&0.720804 \\
     3 & 7.93831 & 0.0&0.0&-1.0 & 0.0&0.0&1.0 \\
     4 & 7.11526 & 0.0&0.0&-1.0 & 1.0&0.0&0.0 \\
     5 & 0.6222 & 0.0&0.0&-1.0 & 0.0&1.0&0.0 \\
    \hline
  \end{tabular}
  \end{center}
\end{table}

%\ce{LiMnO2} belongs to the $m$ crystallographic point group.
%There are ten independent parameters in the piezoelectric tensor $\A_{LiMnO2}$
%\begin{eqnarray*}
%    & a_{111}=2.34136,  \quad a_{122}=0.06249, \quad a_{133}=0.50554, \quad a_{113}=-0.70406, & \\
%    & a_{223}=-0.40829, \quad a_{212}=0.76378, \quad a_{311}=0.33075, \quad a_{322}=-0.23931, & \\
%    & a_{333}=-0.02956, \quad\text{ and }\quad a_{313}=0.21625. &
%\end{eqnarray*}
%Other elements of the piezoelectric tensor is zero.
%Positive C-eigenvalue of $\A_{LiMnO2}$
%and associated C-eigenvectors are reported in Table \ref{C-eigen-LiMnO2}.
%
%\begin{table}[bph]
%  \begin{center}
%  \caption{C-eigenvalues and associated C-eigenvectors of
%  the piezoelectric tensor of $\alpha$-quartz.}\label{C-eigen-LiMnO2}
%  \begin{tabular}{c|c|rrr|rrr}
%    \hline
%    No. & $\lambda$ & \multicolumn{3}{c|}{$\x^\T$} & \multicolumn{3}{c}{$\y^\T$}  \\ \hline
%     1 & 2.5855 & 0.997928&0.&0.0643332 & 0.949448&0.&-0.313925 \\
%     2 & 0.294235 & 0.925346&0.&0.379124 & 0.274158&0.&0.961685 \\
%     3 & 0.247334 & 0.252654&0.&-0.967557 & 0.&1.&0. \\
%     4 & 0.170592 & 0.870132&0.121659&-0.477566 & 0.264222&-0.808591&0.525706 \\
%     5 & 0.170592 & 0.870132&-0.121659&-0.477566 & 0.264222&0.808591&0.525706 \\
%    \hline
%  \end{tabular}
%  \end{center}
%\end{table}

\textbf{Example 6.}
The compound \ce{LiBiB2O5} belongs to the $2$ crystallographic point group \cite{dCGAP-15}.
There are eight independent parameters in the piezoelectric tensor $\A_{LiBiB2O5}$
\begin{eqnarray*}
    & a_{123}=2.35682, \quad a_{112}=0.34929, \qquad a_{211}=0.16101, \qquad a_{222}=0.12562, & \\
    & a_{233}=0.1361, \qquad a_{213}=-0.05587, \qquad a_{323}=6.91074, \quad\text{ and }\quad a_{312}=2.57812. &
\end{eqnarray*}
Other elements of the piezoelectric tensor are zeros.
Positive C-eigenvalue of $\A_{LiBiB2O5}$
and associated C-eigenvectors are reported in Table \ref{C-eigen-LiBiB2O5}.

\begin{table}[bph]
  \begin{center}
  \caption{Positive C-eigenvalues of the piezoelectric tensor of \ce{LiBiB2O5}.}\label{C-eigen-LiBiB2O5}
  \begin{tabular}{c|c|rrr|rrr}
    \hline
    No. & $\lambda$ & \multicolumn{3}{c|}{$\x^\T$} & \multicolumn{3}{c}{$\y^\T$}  \\ \hline
     1 & 7.73762 & 0.302351&0.0148322&0.953081 & 0.234203&0.707114&0.667187 \\
     2 & 7.73762 & -0.302351&0.0148322&-0.953081 & 0.234203&-0.707114&0.667187 \\
     3 & 0.499616 & 0.902379&0.320695&-0.287865 & 0.675213&-0.698513&-0.236998 \\
     4 & 0.499616 & -0.902379&0.320695&0.287865 & 0.675213&0.698513&-0.236998 \\
     5 & 0.205796 & 0.0&1.0&0.0 & 0.780252&0.0&-0.625465 \\
     6 & 0.12562 & 0.0&1.0&0.0 & 0.0&1.0&0.0 \\
     7 & 0.0913135 & 0.0&1.0&0.0 & 0.625465&0.0&0.780252 \\
    \hline
  \end{tabular}
  \end{center}
\end{table}

\textbf{Example 7.}
The compound \ce{KBi2F7} belongs to the $1$ crystallographic point group \cite{dCGAP-15}.
There are eighteen independent parameters in the piezoelectric tensor $\A_{KBi2F7}$
\begin{eqnarray*}
    & a_{111}=12.64393, \quad a_{122}=1.08802, \quad a_{133}=4.14350, \quad a_{123}=1.59052, & \\
    & a_{113}=1.96801, \quad a_{112}=0.22465, \quad a_{211}=2.59187, \quad a_{222}=0.08263, & \\
    & a_{233}=0.81041, \quad a_{223}=0.51165, \quad a_{213}=0.71432, \quad a_{212}=0.10570, & \\
    & a_{311}=1.51254, \quad a_{322}=0.68235, \quad a_{333}=-0.23019, \quad a_{323}=0.19013, & \\
    & a_{313}=0.39030, \quad\text{ and }\quad a_{312}=0.08381. &
\end{eqnarray*}
Positive C-eigenvalue of $\A_{KBi2F7}$
and associated C-eigenvectors are reported in Table \ref{C-eigen-KBi2F7}.

\begin{table}[tbph]
  \begin{center}
  \caption{Positive C-eigenvalues of the piezoelectric tensor of \ce{KBi2F7}.}\label{C-eigen-KBi2F7}
  \begin{tabular}{c|c|rrr|rrr}
    \hline
    No. & $\lambda$ & \multicolumn{3}{c|}{$\x^\T$} & \multicolumn{3}{c}{$\y^\T$}  \\ \hline
     1 & 13.5021 & 0.970501&0.209737&0.118907 & 0.972258&0.0506481&0.228363 \\
     2 & 4.46957 & 0.981961&0.189047&-0.00361752 & 0.22771&-0.414908&-0.880908 \\
     3 & 0.544863 & 0.759805&-0.368785&0.535439 & 0.0616756&0.870474&-0.488334 \\
    \hline
  \end{tabular}
  \end{center}
\end{table}

\textbf{Example 8.}
The compound \ce{BaNiO3} belongs to the $6$ crystallographic point group \cite{dCGAP-15}.
There are three independent parameters in the piezoelectric tensor $\A_{BaNiO3}$
\begin{equation*}
    a_{113}=a_{223}=0.038385, \quad a_{311}=a_{322}=6.89822, \quad\text{ and }\quad a_{333}=27.4628.
\end{equation*}
Other elements of the piezoelectric tensor are zeros.
There are infinite C-eigenvalues of $\A_{BaNiO3}$. We report
positive C-eigenvalue of $\A_{BaNiO3}$
and associated C-eigenvectors in Table \ref{C-eigen-BaNiO3}.

\begin{table}[tbph]
  \begin{center}
  \caption{Positive C-eigenvalues of the piezoelectric tensor of \ce{BaNiO3}.}\label{C-eigen-BaNiO3}
  \begin{tabular}{c|c|rrr|rrr}
    \hline
    No. & $\lambda$ & \multicolumn{3}{c|}{$\x^\T$} & \multicolumn{3}{c}{$\y^\T$}  \\ \hline
     1 & 27.4628 & 0.0&0.0&1.0 & 0.0&0.0&1.0 \\
     2 & 6.89822 & 0.0&0.0&1.0 & $y_1$&$\pm\sqrt{1-y_2^2}$&0.0 \\
    \hline
  \end{tabular}
  \end{center}
\end{table}

\section{Difference from Matrix Singular Values}

An $n$-by-$n$ symmetric matrix $S=[s_{ij}]$ contains $\frac{n(n+1)}{2}$ independent elements.
Hence we may vectorize $S$ as a vector
\begin{equation*}
    \vt{vec}(S) = (s_{11},s_{22},\dots,s_{nn},\sqrt{2}s_{(n-1)n},\dots,\sqrt{2}s_{12})^\T
      \in\REAL^{\frac{n(n+1)}{2}}.
\end{equation*}
Here, we equip off-diagonal elements of $S$ with coefficient $\sqrt{2}$,
while diagonal elements of $S$ are with coefficient $1$.
Hence, we have $\|S\|_F=\|\vt{vec}(S)\|_2$.

Let $\A\in\REAL^{n\times n\times n}$ be a piezoelectric-type tensor
that contains $\frac{n^2(n+1)}{2}$ independent elements.
Owing to the symmetry of latter two indices of $\A$,
we could represent each symmetric slice-matrix as a vector.
Collecting these vectors, we obtain an $n$-by-$\frac{n(n+1)}{2}$ matrix
\begin{equation}
M(\A) = \left(
  \begin{array}{ccccccc}
    a_{111} & a_{122} & \cdots & a_{1nn} & \sqrt{2}a_{1(n-1)n} & \cdots & \sqrt{2}a_{112} \\
    \vdots  & \vdots  & \vdots & \vdots  & \vdots       & \vdots & \vdots   \\
    a_{n11} & a_{n22} & \cdots & a_{nnn} & \sqrt{2}a_{n(n-1)n} & \cdots & \sqrt{2}a_{n12} \\
  \end{array}
\right).
\end{equation}
Here, each row of the above matrix corresponds to a symmetric slice-matrix of the piezoelectric tensor.
Let $\y=(y_1,\dots,y_n)^\T\in\REAL^n$.
By some calculations, we find that
$$ \A\y\y=M(\A)\vt{vec}(\y\y^\T). $$

\begin{Theorem}
  Let $\lambda^*$ and $\mu^*$ be the largest C-eigenvalue of a piezoelectric-type tensor $\A$
  and the largest singular value of the matrix $M(\A)$, respectively.
  Then,
  \begin{equation}\label{eig-sin}
    \lambda^* \leq \mu^*.
  \end{equation}
  In the above inequality,  strict inequality may hold in some cases.
\end{Theorem}
\begin{proof}
  From Theorem 2.3 (c), $\lambda^*$ is the optimal objective value of
  the following maximization problem
  $$ \lambda^* = \max~\{\x\A\y\y : \|\x\|_2=1, \|\y\|_2=1 \}. $$
  We denote the corresponding optimal solution as $(\x^*,\y^*)$.
  Obviously, $\|\vt{vec}(\y^*{\y^*}^\T)\|_2=\|\y^*{\y^*}^\T\|_F\\=\|\y^*\|_2^2=1$.

  By matrix theory, $\mu^*$ is the optimal objective value of
  $$ \mu^* = \max~\{\x^\T M(\A)\z : \|\x\|_2=1, \|\z\|_2=1 \}. $$
  It is easy to see that $(\x^*,\vt{vec}(\y^*{\y^*}^\T))$ is a feasible of this problem.
  Hence, $$ \mu^* \geq {\x^*}^\T M(\A)\vt{vec}(\y^*{\y^*}^\T) = \x^*\A\y^*\y^*=\lambda^*. $$

  \medskip
We now give an example for which strict inequality holds in \eqref{eig-sin}.
Consider a piezoelectric tensor $\A_{0}$ with dimension $2$.
There are only two nonzero elements in $\A_{0}$: $$ a_{112}=a_{222}=1. $$
Hence,
\begin{equation}
M(\A_{0}) = \left(
  \begin{array}{ccccccc}
    0 & 0 & \sqrt{2} \\
    0 & 1 & 0  \\
  \end{array}
\right).
\end{equation}
By some calculation, the largest singular value of $M(\A_{0})$ is $\sqrt{2}$,
and the associated singular vectors are $(1,0)^\T$ and $(0,0,1)^\T$.

Then, we turn to C-eigenvalues of $\A_{0}$.
The system \eqref{C-eig} reduces to
\begin{subequations}
\begin{align}
  2y_1y_2 &= \lambda x_1, \label{E1-a}\\
   y_2^2  &= \lambda x_2, \label{E1-b}\\
         x_1y_2 &= \lambda y_1, \label{E1-c}\\
  x_1y_1+x_2y_2 &= \lambda y_2, \label{E1-d}\\
  x_1^2+x_2^2 &= 1, \label{E1-e}\\
  y_1^2+y_2^2 &= 1. \label{E1-f}
\end{align}
\end{subequations}
If $\lambda=0$, we know $y_2=0$ by \eqref{E1-b}.
From \eqref{E1-f}, we have $y_1=\pm 1$.
Using \eqref{E1-d}, we get $x_1=0$.
By \eqref{E1-e}, we obtain $x_2=\pm 1$.
Hence, we obtain a group of solution
\begin{equation*}
  \lambda_1^*=0,  \qquad \x_1^*=(0,1)^\T, \qquad \y_1^*=(1,0)^\T.
\end{equation*}

In the case of $\lambda\neq 0$ and $y_2=0$, we get $y_1=0$ by \eqref{E1-c},
which contradicts \eqref{E1-f}.

Now, we consider the case that $\lambda\neq 0$ and $y_2\neq 0$.
By \eqref{E1-b}, we know $x_2\neq 0$.
From \eqref{E1-a} and \eqref{E1-b}, we have $$\frac{x_1}{x_2} = \frac{2y_1}{y_2}\equiv t.$$
Then, $x_1=tx_2$ and $y_1=\frac{t}{2}y_2$.
From \eqref{E1-c} and \eqref{E1-d}, we know
$$ \frac{t}{2} = \frac{y_1}{y_2} = \frac{x_1y_2}{x_1y_1+x_2y_2}
= \frac{tx_2y_2}{\frac{t^2}{2}x_2y_2+x_2y_2} =\frac{2t}{t^2+2}. $$
Solving this equation in $t$, we get
$$ t_1=\sqrt{2}, \qquad t_2=-\sqrt{2}, \qquad\text{ and }\qquad t_3=0. $$
If $t=\sqrt{2}$, by $x_1=\sqrt{2}x_2$ and \eqref{E1-e}, we know
$x_1=\pm\frac{\sqrt{2}}{\sqrt{3}}$ and $x_2=\pm\frac{1}{\sqrt{3}}$.
By $y_1=\frac{t}{2}y_2$ and \eqref{E1-f}, we have
$y_1=\pm\frac{1}{\sqrt{3}}$ and $y_2=\pm\frac{\sqrt{2}}{\sqrt{3}}$.
% From Theorem 2.3 (b), plus and minus signs are all accessible.
% Here, we choose $y_1=\frac{1}{\sqrt{3}}$ and $y_2=\frac{\sqrt{2}}{\sqrt{3}}$.
Then, by \eqref{E1-b}, we know $\lambda=\pm \frac{2}{\sqrt{3}}$.
In sum, we have a solution
\begin{equation*}
  \lambda_2^*=\frac{2}{\sqrt{3}},  \qquad \x_2^*=\left(\frac{\sqrt{2}}{\sqrt{3}},\frac{1}{\sqrt{3}}\right)^\T,
     \qquad \y_2^*=\left(\frac{1}{\sqrt{3}},\frac{\sqrt{2}}{\sqrt{3}}\right)^\T.
\end{equation*}

Using similar discussions, we get two more solutions
\begin{equation*}
  \lambda_3^*=\frac{2}{\sqrt{3}}, \qquad \x_3^*=\left(-\frac{\sqrt{2}}{\sqrt{3}},\frac{1}{\sqrt{3}}\right)^\T,
    \qquad \y_3^*=\left(\frac{1}{\sqrt{3}},-\frac{\sqrt{2}}{\sqrt{3}}\right)^\T,
\end{equation*}
and
\begin{equation*}
  \lambda_4^*=1, \qquad \x_4^*=(0,1)^\T, \qquad  \y_4^*=(0,1)^\T.
\end{equation*}
Hence, the largest C-eigenvalue of $\A_{0}$ is $\frac{2}{\sqrt{3}}$.
Obviously, $\frac{2}{\sqrt{3}}<\sqrt{2}$. Hence, strict inequality in \eqref{eig-sin} holds.
  The proof is complete.
\end{proof}

% The example in the proof of the above theorem shows that the C-eigenvalue theory of piezoelectric tensors cannot be replaced by matrix singular value theory.

\section{Summary}

We defined C-eigenvalues and C-eigenvectors for a piezoelectric tensor,
where ``C'' names after Curie brothers.
The existence of C-eigenvalues and its invariance under orthonormal transformations
were also addressed.
We argued that for a piezoelectric tensor,
the largest C-eigenvalue corresponds to the electric displacement vector with
the largest $2$-norm in the piezoelectric effect under unit uniaxial stress,
and the strain tensor with the largest $2$-norm in the converse piezoelectric effect
under unit electric field vector.
Thus, C-eigenvalues and C-eigenvectors have concrete physical meanings in
piezoelectric effect and converse piezoelectric effect.
Finally, we present C-eigenvalues and associated C-eigenvectors of
piezoelectric tensors for various crystal classes.

\end{document}